\documentclass[11pt]{article}
\usepackage[utf8]{inputenc}
\usepackage[letterpaper]{geometry}
\usepackage{ltexpprt}

\usepackage[english]{babel}

\usepackage{amsmath, amssymb, amsfonts}
\usepackage[normalem]{ulem}
\newtheorem{observation}{Observation}

\usepackage{float,graphicx}

\usepackage{pdflscape}
\usepackage{xcolor}
\usepackage{hyperref}
\hypersetup{
	colorlinks,
	linkcolor={green!70!black},%
	citecolor={blue!70!black},%
	urlcolor={blue!80!black}
}

\renewcommand{\O}{\mathcal{O}}

\newcommand{\avg}{\displaystyle \mathop{\mathbb{E}}}
\newcommand{\Perm}{\mathbf{P}}
\newcommand{\A}{\mathcal{T}}

\newcommand{\N}{\mathbb{N}}

\renewcommand{\deg}{f}
\newcommand{\K}{{\cal K}}

\newcommand{\X}{{\cal X}}

\newcommand{\blue}{\color{blue}}

\begin{document}  
\title{A Tail Estimate with Exponential Decay for the Randomized Incremental Construction of Search Structures\footnote{This work was supported under the Australian Research Council Discovery Projects funding scheme (project number DP180102870).
The full version of the paper can be accessed at \protect\url{https://arxiv.org/abs/2101.04914}}
}

\author{    Joachim Gudmundsson\thanks{University of Sydney, Australia. \texttt{joachim.gudmundsson@sydney.edu.au}}  \and     Martin P. Seybold\thanks{University of Sydney, Australia.  \texttt{mpseybold@gmail.com}}}
\date{}

\maketitle

\begin{abstract}
The Randomized Incremental Construction (RIC) of search DAGs for point location in planar subdivisions, nearest-neighbor search in 2D points, and extreme point search in 3D convex hulls, are well known to take ${\cal O}(n \log n)$ expected time for structures of ${\cal O}(n)$ expected size.
Moreover, searching takes w.h.p.~${\cal O}(\log n)$ comparisons in the first and w.h.p.~${\cal O}(\log^2 n)$ comparisons in the latter two DAGs.
However, the expected depth of the DAGs and high probability bounds for their size are unknown. 

Using a novel analysis technique, we show that the three DAGs have w.h.p.
i)   a size  of ${\cal O}(n)$, 
ii)  a depth of ${\cal O}(\log n)$, and
iii) a construction time of ${\cal O}(n \log n)$.
One application of these new and improved results are \emph{remarkably simple} Las Vegas verifiers to obtain search DAGs with optimal worst-case bounds.
This positively answers the conjectured logarithmic search cost in the DAG of Delaunay triangulations [Guibas et al.; ICALP 1990] and a conjecture on the depth of the DAG of Trapezoidal subdivisions [Hemmer et al.; ESA 2012].
\end{abstract}


\section{Introduction}
The Randomized Incremental Construction (RIC) is one of the most successful and influential paradigms in the design of algorithms and data structures.
Its simplicity makes the method particularly useful for many, seemingly different problems that ask to compute a defined structure for a given set of objects.
The idea is to first permute all $n$ objects, uniformly at random, before inserting them, one at a time, in an initially empty structure under this order.
Treaps~\cite{SeidelA96,Vuillemin80} are a 1D example of history-based RIC that also demonstrates the algorithmic use of high probability bounds and rebuilding to maintain worst-case guarantees.

A landmark problem for RIC is computing a planar subdivision, called trapezoidation, that is induced by a set of $n$ line segments~\cite{BergCompGeo, MulmuleyBook}.
Every trapezoidation contains ${\O(n + k)}$ faces, where $k$ is the number of intersections.
Clarkson and Shor~\cite{ClarksonS89} gave the conflict-graph RIC, 
Mulmuley~\cite{Mulmuley90} gave the conflict-list RIC, and Seidel~\cite{Seidel91} gave the history-based RIC that builds the Trapezoidal Search DAG (TSD) online, each taking ${\O(n \log n + k)}$ expected time. 
The TSD is the history of trapezoidations that are created during the RIC and allows to find the trapezoid, of the current subdivision, that contains a query point.
TSDs have a worst-case size of $\Omega(n^2)$, but their expected size is $\O(n + k)$.
Searching takes w.h.p.~$\O(\log n)$ comparisons and the longest search path (search depth) is also w.h.p. $\O( \log n )$, since there are only $\O(n^2)$ different \emph{search paths} (e.g., \cite[Chapter~6.4]{BergCompGeo}).
Similar to Treaps, TSDs allow fully-dynamic updates such that, after each update, the underlying random permutation is uniformly from those over the current set of segments.
Early algorithms generalize search tree rotations
to abstract, complex structures in order to reuse the point location search and leaf level insertion algorithms~\cite{Mulmuley91a-DynShuffle}.
Simpler search and recursive top-down algorithms were described recently~\cite{BrankovicGRS20}.
The bounds on expected insertion and deletion time of both methods however require that the update entails a non-adversarial, random object.

In contrast to Treaps, fundamental questions about the reliability of RIC runtime, maintenance of small DAGs, certifying logarithmic search costs, and avenues to de-randomization, are still not completely understood.
High probability bounds for the TSD construction time are only known under additional assumptions (cf. Section~\ref{sec:rel-work}) and high probability bounds for space and logarithmic search in the DAG of 2D Delaunay triangulations and of 3D convex hulls~\cite{GuibasKS92} are unknown.

\subsection{Related Work}\label{sec:rel-work}
Guibas et al.~\cite{GuibasKS92} showed that history-based RIC for 2D Delaunay triangulations and $3$D convex hulls takes $\O(n \log n)$ expected time.
Their analysis however reveals nothing about the search comparisons in the query phase (see also~\cite[Section~$9.5$]{BergCompGeo}) and the query bound is w.h.p. $\O(\log^2 n)$~\cite[Theorem~$5.4$]{GuibasKS92}.
The authors state \emph{``We believe that this query time is actually $\O( \log n)$''} in several remarks throughout the paper (e.g. p384, p401, p411).
Works on (dynamization of) the two problems~\cite{Mulmuley91a-DynShuffle, Chan10} also have the same $\O(\log^2 n)$ query bound.

The work of Hemmer et al.~\cite{Hemmer16} shows how to turn the TSDs expected query time into a worst-case bound.
They give two, Las Vegas verifier, algorithms to estimate the \emph{search depth}.
Their exact algorithm runs in $\O(n \log n)$ expected time and their $\O(1)$-approximation runs in $\O( n \log n)$ time.
Their CGAL implementation~\cite{cgal:wfzh-a2-20b} however refrains from these verifiers and simply uses the TSD depth to trigger rebuilds, which is a readily available in RIC.
Clearly the TSD depth is an upper bound, since the (combinatorial) paths are a superset of the search paths.
However, the ratio between depth and search depth is $\Omega( n /\log n )$.
The authors conjecture that TSD depth is $\O( \log n)$ with at least constant probability (see Conjecture~$1$ in \cite{Hemmer12}).
To the best of our knowledge, the expected value of this quantity is still unknown.

The theory developed for RICs lead to a tail bound technique~\cite{MehlhornSW92, ClarksonMS93} that holds as soon as the actual geometric problem under consideration provides a certain boundedness property.
The strongest known tail bound is from Clarkson et al.~\cite[Corollary~26]{ClarksonMS93}, which states the following.
Given a function $M$ such that $M(j)$ upper bounds the size of the structure on $j$ objects.
If $M(j)/j$ is non-decreasing, then, for all $\beta>1$, the probability that the history size exceeds $\beta M(n)$ is at most $(e/\beta)^\beta/e$.
This includes the TSD size for non-crossing segments ($k=0$), but also the DAGs for 3D convex hulls and 2D Delaunay triangulations.
Assuming intersecting segments, Matou\v{s}ek and Seidel~\cite{MatousekS92} show how to use an isoperimetric inequality for permutations to derive a tail bound of $\O(n^{-c})$, given there are at least $k\geq C n \log^{15} n$ many intersections in the input (both constants $c$ and $C$ depend on the deviation threshold $\beta$).
Mehlhorn et al.~\cite{MehlhornSW92} show that the general approach can yield a tail bound of at most $1/e^{\Omega(k/n \log n)}$, given there are at least $k\geq n \log n \log \log \log n$ intersections in the input segments.

Recently, Sen~\cite{Sen19} gave tail estimates for conflict-graph RICs (cf. Chapter~3.4 in \cite{MulmuleyBook}) using Freedman's inequality for Martingales.
The work also shows a lower bound on tail estimates for the runtime, i.e. the total number of conflict-graph modifications, for computing the trapezoidation of non-crossing segments that \emph{rules out} high probability tail bounds~\cite[Section~6]{Sen19}.
In conflict-graph RIC, not only one endpoint per segment is maintained in conflict lists, but edges in a bipartite graph, over existing trapezoids and uninserted segments, that contain an edge if and only if the geometric objects intersect (see Appendix and Figure~4 in \cite{Sen19}).
Hence the lower bound construction only applies to conflict-graph RIC and does not translate to the history-based RIC. 

\begin{table}[t] \centering
	\setlength{\tabcolsep}{6pt} %
	\begin{tabular}{|l r|c|l|l |} \hline 
		Technique & & Bound & With Prob. $\geq$ & Condition \\ 	\hline
		Isoperimetric & \cite{MatousekS92}
		& $\O(n+k)$ 
		& $1-\O(1/n^{c})$
		& $k\geq C n \log^{15} n$ \\ 
		Hoeffding & \cite{MehlhornSW92} 
		& $\O(n+k)$
		& $1-1/e^{\Omega(k/n \log n)}$ 
		& $k\geq n \log n \log \log \log n$ \\
		Freedman & \cite[Lem.12]{Sen19} %
		& $\O(n+k)$
		& $1-1/e^{k/n \alpha(n)}$ 
		& $k\geq n \log n$ \\ 		
		Events & \color{red}\sout{Section~\ref{sec:tsd-size-crossing}}
		& $\O(n+k)$
		& $1 -1/e^{n+k}$
		& \\ \hline
		Hoeffding &\cite{ClarksonMS93} 
		& $\O(\beta  n)$
		& $1-(e/\beta)^\beta/e$ 
		& \\ 
		Pairwise Events & \color{red}\sout{Section~\ref{sec:tailbound}}
		& $\O(n)$
		& $1 -1/e^n$
		& \\
		\hline
	\end{tabular}
	\caption{Tail bounds for the history size of TSDs on $n$ segments. $k$ denotes the number of intersections and $\alpha(n)$ the inverse of Ackermann's function.}\label{tab:tailbounds}
\end{table}

\subsection{Contribution}
We introduce a new and direct technique to analyze the history size that fully abstracts from the geometric problem to Pairwise Events of object adjacency.
Using a matrix property of the events enables an inductive Chernoff argument, despite the lack of full independence.
The main result in Section~\ref{sec:tailbound} is a \emph{much sharper} tail estimate for the TSD size (see Table~\ref{tab:tailbounds}).
This complements the known high probability bound for the point location cost and shows that TSD construction takes w.h.p. $\O(n \log n)$ time.
Moreover, maintaining a TSD size of $\O(n)$, in the static and dynamic setting, merely adds an expected rebuild cost of $\O(1)$.

Unlike geometric Backward Analysis of the query cost, we deal with union bounds over exponential domains, inherent to our combinatorial approach, in Section~\ref{sec:path-length}.
We derive (inverse) polynomial probability bounds and show that the exponential union bound adds up to a high probability bound.
The main result in this section is that the TSD has w.h.p. $\O( \log n)$ depth, and thus confirms the conjecture of Hemmer et al.~\cite{Hemmer12,Hemmer16} with a \emph{substantially stronger} bound.

In Section~\ref{ssec:dag-triangulations}, we show that our technique allows to obtain identical bounds for size, depth, and runtime for the search DAGs of 2D Delaunay triangulations and 3D convex hulls.
This improvement answers the conjectured query time bound of Guibas et al.~\cite{GuibasKS92} affirmatively.
Additional bookkeeping during RIC allows us to track our (high probability) bound on the maximum cost of edge weighted root-to-leaf paths in the DAGs.
Hence space can be made $\O(n)$ and query time can be made $\O(\log n)$ worst-case bounds with rebuilding.

In Section~\ref{sec:tsd-size-crossing}, we extend the technique to yield improved bounds for a history-based RIC with $\omega(n^2)$ worst-case size.
We show that the TSD size is $\O(n+k)$ with probability at least $1-1/e^{n+k}$, where $k\geq 0$ is the number of intersections of the segments (see Table~\ref{tab:tailbounds}).

\section{Recap: Trapezoidal Search DAGs} \label{sec:recap-tsd}

For a set $S$ of $n$ segments in the plane, we denote by $\K(S) \subseteq \binom{S}{2}$ the set of crossings and ${k=|\K(S)|}$. %
We identify the permutations over $S$ with the set of bijective mappings to $\{1,\ldots,n\}$, i.e.
 $\Perm(S)=\{ {\pi : S \to \{1,\ldots,n\}}~|~\pi \text{ bijective}\}$.
Integer $\pi(s)$ is called the \emph{priority} of the segment $s$.

An implicit, infinitesimal shear transformation allows to assume, without loss of generality, that all distinct end and intersection points have different $x$-coordinates (e.g. Chapter 6.3 in~\cite{BergCompGeo}).
Trapezoidation $\A(S)$ is defined by emitting two vertical rays (in negative and positive $y$-direction) from each end and intersection point until the ray meets the first segment or the bounding rectangle (see Figure~\ref{fig:example-trapezoidations}).
To simplify presentation, we also implicitly move common end and intersection points infinitesimally along their segment, towards their interior.
This gives that segments have \emph{no} points in common, though there may exist additional spatially empty trapezoids in $\A(S)$.
We identify $\A(S)$ with the set of faces in this decomposition of the plane.
Elements in $\A(S)$ are trapezoids with four boundaries that are defined by at least one and at most four segments of $S$ (see Figure~\ref{fig:example-trapezoidations}).
Note that boundaries of the trapezoids in $\A(S)$ are solely determined by the set of segments $S$, irrespective of the permutation.
We will need the following notations.
Let $\gamma>1$ be the smallest constant%
\footnote{Counting trapezoids in a x-sweep shows $|\A(S)|= 1+3(n+k)$ (see \cite[p.127]{BergCompGeo} and \cite[Section~3]{Seidel1993}).}
such that
$|\A(S)| \leq (n+k)\gamma$
holds for any $S$. %
For a segment $s \in S$, let $\deg(s,S) = {\{ \Delta \in \A(S):\Delta \text{ is bounded by } s\}}$ denote the set of faces that are bounded by $s$ (i.e. {\it top, bot, left,} or {\it right}).
Let $s_i=\pi^{-1}(i)$ be the priority $i$ segment and let $S_{\leq k}=\{s_1, \ldots, s_k\}$.

The expected size of the TSD is typically analyzed by considering $\sum_{j=1}^n D_j$ where the random variable $D_j := \left| \deg(s_j, S_{\leq j})\right|$ denotes the number of faces that are created by inserting $s_j$ into trapezoidation $\A(S_{\leq j-1})$, equivalently that are removed by deleting $s_j$ from $\A(S_{\leq j})$ (see Figure~\ref{fig:example-tsd}).

\begin{figure}
	\centering
	\includegraphics[width=.4\columnwidth,page=1]{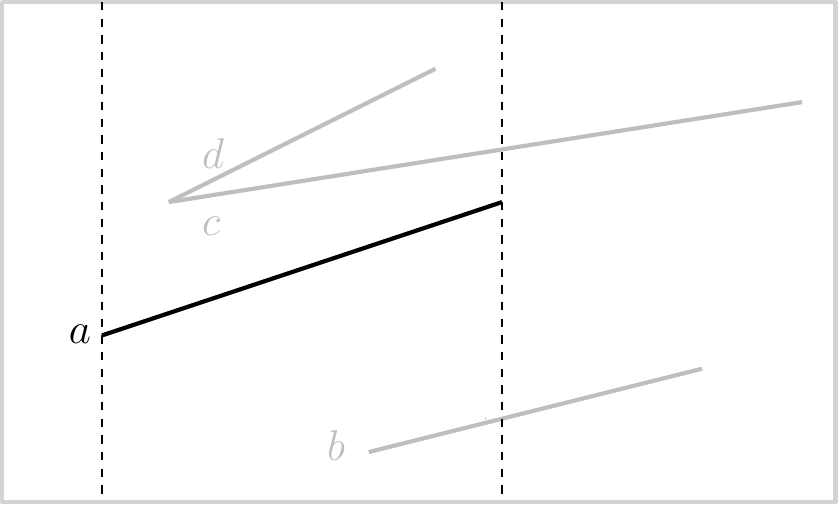}
	\includegraphics[width=.4\columnwidth,page=2]{./figs/tst-vs-tsd.pdf}\\
	\includegraphics[width=.4\columnwidth,page=3]{./figs/tst-vs-tsd.pdf}
	\includegraphics[width=.4\columnwidth,page=4]{./figs/tst-vs-tsd.pdf}
	\caption{
		Trapezoidations over the segments $S=\{a,b,c,d\}$, with $a=(a.l,a.r)$, $b=(b.l,b.r)$, $c=(c.l,c.r)$, and $d=(d.l,d.r)\}$, where $c.l=d.l$ is a common endpoint.
		$\A(\{a\})$, $\A(\{a,b\})$, $\A(\{a,b,c\})$, and $\A(\{a,b,c,d\})$ have $4$, $7$, $10$, and $13$ faces respectively (cf.  Figure~\ref{fig:example-tsd}).
	}
	\label{fig:example-trapezoidations}
\end{figure}
\begin{figure}
	\centering
    \includegraphics[width=.6\columnwidth]{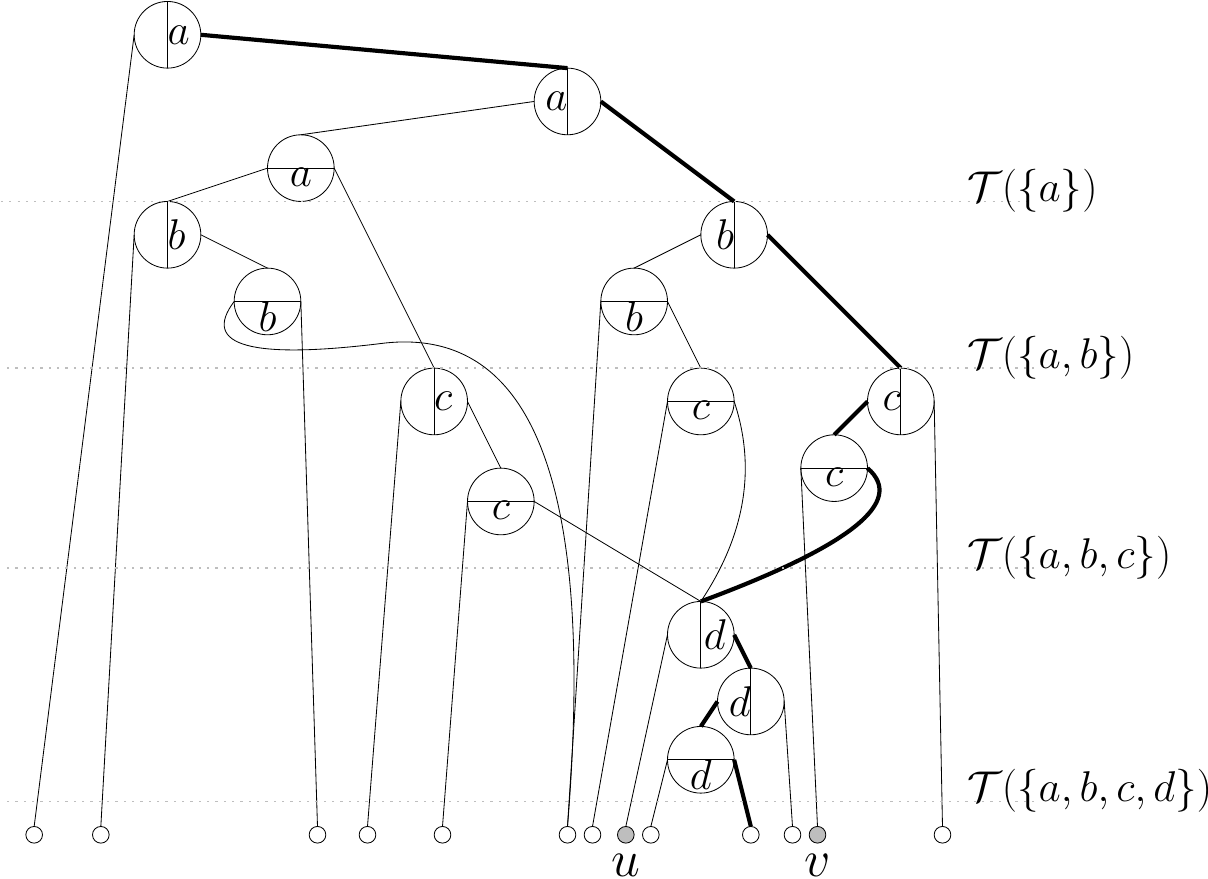}
	\caption[TSD]{TSD for the history of trapezoidations under permutation $\pi = \bigl(\begin{smallmatrix}
		a & b 	& c & d\\
		1 	& 2 	& 3 & 4 
		\end{smallmatrix}\bigr)$ 
		from Figure~\ref{fig:example-trapezoidations}.
		TSD node $v$ corresponds to the trapezoid $\Delta(v)$, which has the boundaries $\textit{top}(\Delta(v))=c$, $\textit{bot}(\Delta(v))=b$, $\textit{left}(\Delta(v))=a.r$, and $\textit{right}(\Delta(v))=b.r$ and the spatially empty $\Delta(u)$ is due to common endpoint {\it left}$(\Delta(u))=c.l=d.l=\textit{right}(\Delta(u))$.
		The path with heavy line width is not a search path, since $d.r$ is left of $a.r$.}
	\label{fig:example-tsd}
\end{figure}

Classic Backward Analysis~\cite[p.~$136$]{BergCompGeo} in this context is the following argument.
Let $S' \subseteq S$ be a fixed subset of $j$ segments, then
\begin{align*}  %
\avg_{\Perm(S)} \Big[ D_j \Big| S_{\leq j} = S'\Big] 
&=
\frac{1}{j} \sum_{s\in S'} \sum_{\Delta \in \A(S')} \chi \big( \Delta \in \deg(s,S') \big) 
\quad %
\leq 4 \frac{|\A(S')|}{j}
\end{align*}
where the binary indicator variable $\chi(\Delta \in \deg(s,S'))$ is $1$ if and only if the trapezoid $\Delta$ is bounded by segment $s$.
The equality is due to that every segment in $S'$ is equally likely to be picked for $s_j$.

For $k=0$, we have that $|\A(S')|\leq \gamma j$, regardless of the actual set $S'$, and
$\avg [D_j] \leq 4\gamma$ holds unconditionally for each step $j$.
For $k > 0$ one observes that any given crossing in $\K(S)$ is present in $S_{\leq j}$ with probability $\frac{j(j-1)}{n(n-1)}$, hence summing over the $k$ crossings gives that 
$\avg|\A(S_{\leq j})| \leq \gamma(j + k j^2/n^2)$ and
$\avg[D_j] \leq 4\gamma(1 + kj/n^2)$ and thus $\sum_j \avg[D_j] \leq 4\gamma(n + k)$.
Replacing $k$ in this bound with the number of intersection points incurs a more technical argument (see \cite[p.~46]{Seidel1993}).

Since the destruction of a face (of a leaf node) creates at most three DAG nodes, the expected number of TSD nodes is at most $12 \gamma (n+k)$.

\section{Stronger tail bounds using Pairwise Events} \label{sec:tailbound}
Let $S$ be a set of $n$ non-crossing segments throughout this section.
Segments $s$ and $s'$ are called \emph{adjacent} in $S$ if there is a face $\Delta \in \A(S)$ that is \emph{incident} to both, i.e. both define some part of the boundary of $\Delta$.
We define for each $1 \leq i < j \leq n$ an event, i.e. a binary random variable, $X_{i,j} : \Perm(S) \to \{0,1\}$ that occurs if and only if $s_i$ and $s_j$ are adjacent in $\A(S_{\leq j})$.
That is
\[
X_{i,j} = \begin{cases}
1 & \text{if } \deg(s_j, S_{\leq j}) \text{ contains a trapezoid bounded by } s_i \\
0 & \text{otherwise}
\end{cases}\quad.
\]
To simplify presentation, we place the events in a lower triangle matrix and call the set $r(j):=\{X_{i,j}:1\leq i < j\}$ the events of row $j$ and the set $c(i):=\{X_{i,j}:i<j\leq n\}$ the events of column $i$.

Imagine that the random permutation is built backwards, i.e. by successively choosing one of the remaining elements uniformly at random to assign the largest available priority value.
For every step $j$ at least one of the row events occurs, i.e. $0 < \sum_{i<j} X_{i,j} < j$, since at least one trapezoid is destroyed in step $j$ and the exact probability of the events $r(j)$ depends on the geometry of the segments $S_{\leq j}$.

Consider the events in row $j$.
Conditioned on the random permutation starting with set $S'= S\setminus\{s_{j+1},\ldots,s_n\}$, 
the experiment chooses $s \in S'$ uniformly at random and assigns the priority value $j$ to it.
Clearly the number of occurring (row) events depends on which segment of $S'$ is picked as $s_j$, as this determines the value $\deg(s_j,S')$.
Note that the choice of $s_j$ also fixes a partition of $S' =\{s_j\}\cup  A \cup N$ into those segments that are and aren't adjacent to $s_j$, the sets $A$ and $N$ respectively.
This defines a partition $S_{\leq j}=\{s_j\}\cup A_j \cup N_j$ in every backward step $j$.
Eventually $s_i$ is picked from $S_{\leq i} $, which determines the outcomes of \emph{all events} in $c(i)$.
I.e. when $s_i$ is picked from $S_{\leq i}$, the objects in the set are multicolored ($A_j$ or $N_j$ for each $j>i$) and $X_{i,j}$ occurs if and only if the pick has the respective color $A_j$.

\begin{table}\centering
	$
	X (\pi)=
	\begin{array}{|c|c|c|c|} 	\hline
	~&~&~&~~		\\ \hline
	1&~&~&~~ 	\\ \hline
	1&1&~&~~ 	\\ \hline
	0&0&1&~~		\\ \hline
	\end{array}
	\quad
	\begin{array}{ll|r}
	~&~					   	 & D_1=4 \\ 
	A_2=\{a\}	&N_2=\{\}  	 & D_2=5 \\ 
	A_3=\{a,b\}	&N_3=\{\}  	 & D_3=6 \\ 
	A_4=\{c\}	&N_4=\{a,b\} & D_4=4
	\end{array}$
	\caption{Outcome of the pairwise events and the partitions for segments $S=\{a,b,c,d\}$ and order $\pi$ from Figures~\ref{fig:example-trapezoidations} and \ref{fig:example-tsd}.}
	\label{table:events}
\end{table}

This shows for the event probability that
\begin{align} \label{eq:conditioning-uninformative}
\avg[X_{i,j}|s_{i+1},\ldots,s_n] = \avg[X_{i,j}| Y, s_{i+1},\ldots,s_n]
\end{align}
for every $Y \in c(i')$ with $i'>i$.
See Table~\ref{table:events} for an example.

Moreover, we have, for every $t>0$, the two equations
\begin{align} 
\avg[ \prod_{j>1}\exp(tX_{1,j})| s_2,\ldots,s_n] &=
\prod_{j>1} \avg[\exp(tX_{1,j})| s_2,\ldots,s_n]
\label{eq:deter-events-c1}
\\
\avg[ \prod_{j>i}\exp(tX_{i,j})| s_i,\ldots,s_n] &=
\prod_{j>i} \avg[\exp(tX_{i,j})| s_i,\ldots,s_n]~,
\label{eq:deter-events-ci}
\end{align}
where $(s_i,\ldots,s_n)$ denotes the condition that the random permutations have this suffix.
Note that for a set of events $\{B_i\}$ that are either certain or impossible, i.e. ${\normalfont \avg}[B_i]=\Pr[B_i] \in \{0,1\}$, we have that the outcome of each event is identical to its probability and thus  
$
{\normalfont \avg}[ \prod_i\exp(t B_i) ]
=
\prod_i\exp(t {\normalfont \avg}[B_i])
=
\prod_i {\normalfont \avg}[\exp(t B_i)]
$.

There is a close relation between the row events $r(j)$ and the random variable $D_j$.

\begin{lemma} \label{lem:relation-events-destroyed}
For each $\pi \in \Perm(S)$ and $j \geq 2$, we have 
$D_j(\pi)/6 \leq \sum_{i<j} X_{i,j}(\pi) \leq 3 D_j(\pi)$.
\end{lemma}

\begin{proof}
	Let $S' := S_{\leq j}(\pi)$ be the segments with priority at most $j$ in $\pi$.
	Clearly every trapezoid that is incident to $s_j$ is bounded by at most three other segments, which gives the upper bound.
	For the lower bound, we first count those trapezoids of $\deg(s_j,S')$ that have $s_j$ as top or bottom boundary.
	Let $P$ be the set of endpoints that define the vertical boundaries of these trapezoids, excluding the endpoints of $s_j=(q_l,q_r)$.
	Partition $P$ into points $P^+$ above and $P^-$ below $s_j$, which blocks their vertical rays in trapezoidation $\A(S')$.
	Consider the two sets $P^+$ and $P^-$ sorted by their $x$-coordinates.
	
	Between the endpoints of $s_j$, the vertical boundaries of points in $P^+$ can define at most $|P^+|+1$ trapezoids.
	Hence $\deg(s_j,S')$ contains at most $|P|+2$ trapezoids that have $s_j$ on their top or bottom boundary.
	The remaining trapezoids of $\deg(s_j,S')$ are either bounded by $q_l$ or by $q_r$.
	There is at most one trapezoid in $\A(S')$ that has endpoint $q_l$ as right vertical boundary but
	not $s_j$ as bottom or top segment.
	The argument for $q_r$ is symmetric.
	
	Putting the bounds for all cases of trapezoids in $\deg(s_j,S')$ together and using the fact that 
	$|P|\leq 2               \sum_{i<j} X_{i,j}(\pi)$, we have
	\begin{align*}
	D_j(\pi) 
	\leq 
	(2+|P|) + 2
	&\leq
	\Big(2+ 2 \sum_{i<j} X_{i,j}(\pi) \Big) + 
	2
	\leq 6  \sum_{i<j}X_{i,j}(\pi)~.
	\end{align*}
	In the last step we used the fact that $1 \leq \sum_{i<j}X_{i,j}(\pi)$.
	\end{proof}

Note that the respective upper and lower bounds hold for \emph{every} permutation $\pi \in \Perm(S)$.
This shows that the expected number of events that occur in row $j$ is at most $3\avg[D_j]\leq 12\gamma$  (and at least $1$).
Thus ${\normalfont \avg}[\sum_{i,j} X_{i,j}]$ is in the interval $[n-1,12\gamma n]$.

Furthermore, consider the isolated event $X_{i,j}$ in row $j$.
Since the element $s_i$ is picked uniformly at random from the set $S_{\leq j-1}$, we have that its event probability is within the range
\begin{align} \label{eq:harmonic-events}
\frac{1}{j-1} ~\leq~ 
  \avg[X_{i,j}] ~\leq~ \frac{3\avg[D_j]}{j-1} 
  ~\leq~\frac{12\gamma}{j-1}.
\end{align}
Hence the events have roughly Harmonic distribution, i.e. up to bounded multiplicative distortions.

We find it noteworthy that our technique completely captures, with only one lemma, the entire nature of the geometric problem within these (constant) distortion factors of the pairwise events and thus generalizes easily to other RICs (see Section~\ref{sec:conf-spaces}).

However, due to the nature of the incremental selection process, there is a dependence between the events in $r(j)$, e.g. between $X_{i,j}$ and $\{X_{i+1,j},\ldots,X_{j-1,j}\}$.
We circumnavigate this obstacle using conditional expectations in the proof of our tail bound.

\begin{theorem}\label{thm:inductive-chernoff} \color{red}
The random variable $X := \sum_{i,j} X_{i,j}$ has an exponential upper tail, i.e., we have
$\Pr[X \geq T] \leq \exp(\avg[X] - T\ln 2)$ for all $T>0$.
\end{theorem}

\begin{proof}
	To leverage Equation~(\ref{eq:deter-events-c1}) and (\ref{eq:deter-events-ci}) for our events, we regroup the summation terms by column index. 
	Let $t:=\ln 2$ and $C_i := \sum_{Y \in c(i)} Y$ for each $1 \leq i < n$.
	Markov's inequality gives that
	\begin{align}
			\Pr\Big[ \sum_{i<n} C_i \geq T \Big]
		=	
		\Pr\Big[ \exp\big( t \sum_{i<n} C_i \big) \geq 	e^{tT} \Big] 
		\leq  \avg\Big[ \exp\big( t \sum_{i<n} C_i \big) \Big] / e^{tT}~.
		\label{eq:TSD-size-markov}
	\end{align}
	Defining $Q_i:=\exp(tC_1 + \ldots +tC_i)$, we will show by induction that
	${\normalfont \avg}[Q_{n-1}|s_n] \leq \exp({\normalfont \avg}[\sum_{i<n}C_i|s_n] )$ for each $s_n \in S$.
	The condition $(s_i,\ldots,s_n)$ denotes that the permutations $\Perm(S)$ are restricted to those that have this suffix of elements. 
	
	For $i=1$ and each suffix condition $(s_2,\ldots,s_n)$, we have
	\begin{align*}
		\avg[e^{tC_1}|s_2,\ldots,s_n]
		=  	%
		&\avg[\prod_{j=2}^ne^{tX_{1,j}}|s_2,\ldots,s_n] 		\\
		=  	
		& \prod_{j=2}^n \avg[e^{tX_{1,j}}|s_2,\ldots,s_n]
		\\
		=  		%
		&\prod_{j=2}^n \Big( (1-\avg[ X_{1,j}|s_2,\ldots,s_n])e^0 + \avg[ X_{1,j}|s_2,\ldots,s_n] e^t	\Big)\\
		=  		%
		&\prod_{j=2}^n \Big(1+ \underbrace{(e^t-1)}_{=1} \avg[ X_{1,j}|s_2,\ldots,s_n]	\Big)								\\
		\leq 	%
		&\exp\Big( \avg[\sum_{j=2}^n X_{1,j}|s_2,\ldots,s_n]  \Big)
		= \exp\Big( \avg[C_1 |s_2,\ldots,s_n]  \Big)
		~,
	\end{align*}
	where the second equality is due to Equation~(\ref{eq:deter-events-c1}) under the given suffix condition.
	The third equality is due to the definition of expected values, the fourth due to the distributive rule, and the fifth equality due to our choice of $t$.
	The inequality is due to the well known inequality $1+x\leq e^x$.

	For $i>1$ and each condition $(s_{i+1},\ldots,s_n)$, let $S'={S\setminus\{s_{i+1},\ldots,s_n\}}$ and we have
	\begin{align*}
		\avg\Big[ Q_i \Big| s_{i+1},\ldots,s_n   \Big] 
		 =
		&\frac{1}{i} \sum_{ s_i \in S'}
		\avg\Big[ Q_{i-1} \cdot e^{tC_i} \Big
		| s_i, s_{i+1},\ldots, s_n
		\Big]
		\\ =
		&\frac{1}{i} \sum_{ s_i \in S'}
		\avg\Big[
		\underbrace{\avg[Q_{i-1}|c(i), s_i,\ldots,s_n]}_{=\avg[Q_{i-1}|s_i,\ldots,s_n]} 
		e^{tC_i} 
		\Big| s_i,\ldots,s_n   \Big] 	
		\\* \leq
		&\frac{1}{i} \sum_{ s_i \in S'}
		\exp(\avg[ C_1+\ldots+C_{i-1}	| s_i,\ldots,s_n]
		\cdot
		\underbrace{\avg[e^{tC_i} | s_i,\ldots,s_n ]}_{
			\leq \exp(\avg[C_i| s_i,\ldots,s_n])}
		\\
		\leq
		&\frac{1}{i} \sum_{ s_i \in S' }
		\exp(\avg[ C_1+\ldots+C_{i}	| s_{i},\ldots,s_n]
		\\ 
		\color{red} \leq
		& \color{red} \exp(\avg[ C_1+\ldots+C_{i}	| s_{i+1},\ldots,s_n]).
		 \llap{\sout{\phantom{$\leq \exp(\avg[ C_1+\ldots+C_{i}	| s_{i+1},\ldots,s_n]).$}}}
	\end{align*}
	The first equality is due to that every element of $S'$ is equally likely to be picked for $s_{i}$.
	The second equality is due to the `law of total expectation'.
	The third equality is due to a property of our events, see Equation~(\ref{eq:conditioning-uninformative}).
	The resulting terms are bounded by the induction hypothesis and analogously to the case $i=1$, but using Equation~(\ref{eq:deter-events-ci}) for the events $c(i)$ instead.
	{\color{red} \sout{This concludes the induction.
	
	Since
	${\normalfont\avg}[Q_{n-1}]=\frac{1}{n}\sum_{s_n \in S} {\normalfont\avg}[Q_{n-1}|s_n]$,
	we have that
	$ {\normalfont\avg}[Q_{n-1}] \leq \exp\big( {\normalfont \avg}[\sum_{i<n}C_i] \big)$
	and the result follows from (\ref{eq:TSD-size-markov}).}}
\end{proof}

Using the upper bound from Lemma~\ref{lem:relation-events-destroyed}, we have 
${\normalfont \avg}[\sum_{i,j}X_{i,j}] \leq 12\gamma n$.
Since the lower bound of the lemma holds for every permutation, we have for all $T>0$ that 
$\Pr[ T \leq \sum_j D_j ] \leq \Pr[ T \leq 6 \sum_{i,j} X_{i,j} ]$.
Hence, choosing $T = \beta n$ with a sufficiently large constant $\beta$ gives the following result.
\begin{corollary}\label{cor:tsd-size-non-crossing}
	The TSD size is $\O(n)$ with probability at least $1-1/e^n$.
\end{corollary}

This complements the known high probability bound for the point location cost\footnote{Cf.~\cite[Chapter~6.4]{BergCompGeo} and \cite[Lemma~3.1.5 and Theorem~3.1.4]{MulmuleyBook}.} and shows that the TSD of non-crossing segments has, with \emph{very high probability}, size $\O(j)$ after every insertion step $j$.
Since the RIC time for the TSD of non-crossing segments solely entails point location costs and search node creations, we have shown the following statements.

\begin{corollary}
	The Randomized Incremental Construction of a TSD for $n$ non-crossing segments takes w.h.p. $\O(n \log n)$ time.
\end{corollary}
\begin{corollary}
	The TSD size for non-crossing segments can be made deterministic $\O(n)$ with `rebuild if too large' by merely increasing the expected construction time by an additive constant.
\end{corollary}

\section{Depth in the History DAG is w.h.p. logarithmic}     \label{sec:path-length}
For the TSD and other RICs (cf. Section~\ref{sec:conf-spaces}), the ubiquitous argument shows that the \emph{search path} to an arbitrary, but fixed, point has with high probability logarithmic length.
Since there are only $\O(n^2)$ search paths, the high probability bound is strong enough to address each of them in a union bound (see e.g.~\cite[Chapter~6]{BergCompGeo}).
However, a DAG on $n$ vertices of degree at most two may well contain $\Omega(3^{n/2})$ different \emph{combinatorial paths} (see Figure~\ref{fig:many-paths}), hence the same argument cannot be used to obtain a high probability bound.
Our technique allows to derive high probability bounds for this problem and thus gives \emph{remarkably simple} Las Vegas verifiers using the length of a longest combinatorial path.
We first introduce the method for the TSD of non-crossing segments and discuss modifications for the RIC of Delaunay Triangulations and Convex Hull in Section~\ref{sec:conf-spaces}.

Each root-to-leaf path in the TSD visits a sequence of `full region' nodes $(u_1, \ldots, u_m)$, i.e. those nodes whose associated trapezoids $\Delta(u_i)$ are actual faces of the trapezoidation $\A(S_{\leq j})$ for some step $j\in \{1,\ldots,n\}$ (see Figure~\ref{fig:example-tsd}).
The length of the sequence of face transitions, is within a factor of three of the path length since a face destruction inserts at most three edges in the TSD to connect a trapezoid of $\A(S_{\leq i-1})$ with one in $\A(S_{\leq i})$.
We are interested in an upper bound on the number of face-transitions that lead from the trapezoidation $\A(S_{\leq 1})$ to a face of $\A(S_{\leq n})$.

Given such a trapezoid $\Delta$, we call $b(\Delta):=\min\{\pi(s):s \text{ is a boundary of } \Delta\}$ the boundary priority of the trapezoid.
For a sequence of trapezoids 
$
(\Delta_1, \Delta_2, \ldots, \Delta_m)
$
on a root-to-leaf path, the respective sequence of boundary priority values $b(\Delta_\eta)$ is monotonically non-decreasing. 
Let $1=b_0 < b_1 < b_2 < \ldots < b_\ell$ be their sequence of distinct values. %
The number of trapezoids on this path that have the boundary priority value $b_i$ is at most 
\[
\Big| \big\{ \Delta_\eta : b(\Delta_\eta) = b_i \big\} \Big|
\leq 
\sum_{~b_i < j \leq b_{i+1}}  X_{b_i,j}
\quad,
\] 
since the destruction of a trapezoid necessitates that the segment with priority $b_i$ is adjacent to the segment $s_j$ that causes the destruction.
Moreover, event $X_{b_i,b_{i+1}}$ needs to occur if the sequence of priority values stems from a sequence of face transitions in the TSD, i.e. only if $X_{b_i,b_{i+1}}$ occurs we can have a face transition from one with boundary priority $b_i$ to $b_{i+1}$.
We call a sequence of indices \emph{(geometrically) feasible} on the permutation $\pi \in \Perm(S)$ if those specified events occur (e.g.~Figure~\ref{fig:sequence}).
In the example of Table~\ref{table:events}, only the sequences $(1,2), (1,2,3), (1,2,3,4),(1,3)$ and $(1,3,4)$ are feasible on the permutation.
Note that this also means that any feasible sequence $(b_0,\ldots,b_\ell)$, always has at least $\ell$ occurring events in its truncated column sums, that is
\[
\ell \leq 
m \leq \sum_{i=1}^{\ell}  \sum_{~b_{i-1} < j \leq b_{i}}  X_{b_{i-1},j} 
                         +\sum_{~b_\ell < j \leq n}  X_{b_\ell,j} 
\quad.
\]

To simplify notation in this section, let $\delta$ be the smallest integer such that Equation~(\ref{eq:harmonic-events}) turns into
${\normalfont \avg}[X_{i,j}]\leq\frac{12\gamma}{j-1}\leq \delta/j$ for all $j>1$ (i.e. $\delta:=\lceil 12\gamma2\rceil$).
Thus the expected value of the truncated column sum $\mu_{x,y} :={\normalfont \avg}[ \sum_{x < j \leq y}  X_{x,j}]$ is at most
$
\delta \big( H_y - H_x \big)
$,
where $H_n$ denotes the $n$-th Harmonic number.
Summing over the expectation bounds of the sequence's parts, we have for the number of trapezoids $m$ along this sequence that $m/\delta$ is at most
\begin{align}   \label{eq:bound-faces-of-sequence}
    \sum_{i=1}^\ell \Big( H_{b_{i}}-H_{b_{i-1}} \Big)
+   \Big( H_{n}-H_{b_\ell} \Big) 
= H_n -H_{b_0}~.
\end{align}
Note that this bound does not depend on the actual sequence of boundary priority values $\{b_i\}$,
hence the expected number of face transitions of \emph{any} sequence is $\O(\ln n)$.

\subsection{How likely are long combinatorial paths?}

\begin{figure}[]
    \begin{minipage}[t]{.5\textwidth} \centering
        \includegraphics[height=4.5cm]{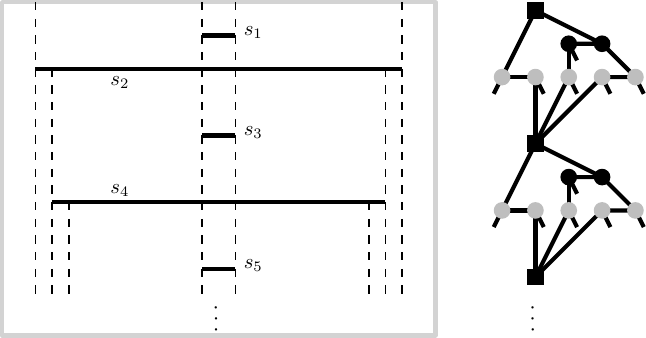}    
        \caption{Insertion order for non-crossing segments that results in a TSD with depth $\Omega(n)$ and $\Omega(3^{n/2})$ root-to-leaf paths.}
        \label{fig:many-paths}
    \end{minipage}\hfill
    \begin{minipage}[t]{.45\textwidth}\centering
        \includegraphics[height=4.5cm]{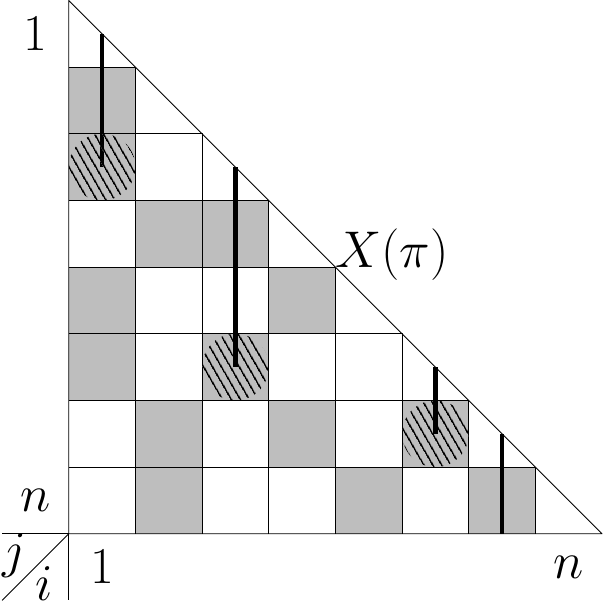}
        \caption{Example of absent events (white) and occurring events (gray). The sequence $\sigma=(1,3,6,7)$ is feasible since $X_{1,3},X_{3,6}$, and $X_{6,7}$, occur (tiled circles).
        A total of six events occur on $\sigma$, i.e. $\sigma(\pi)=6$.}
        \label{fig:sequence}
    \end{minipage}
\end{figure}

We draw the random matrix $C_{i,k}:= \sum_{i<j\leq k} X_{i,j}$ from the probability space.
From the $n\times n$ elements of $C$, the diagonal and upper triangle entries are all zero.
Note that in $C$, the values within a column heavily depend on those of the rows above.
If we find that all matrix elements of $C$ are within a constant factor, say $\beta>1$, of their expected value, we immediately have that the bound in Equation~(\ref{eq:bound-faces-of-sequence}) holds for any sequence of column indices, i.e. \emph{all} TSD paths regardless of their number.
Establishing an upper bound on
$
\Pr\left[  \exists ~ 1\leq i < k \leq n ~:~ C_{i,k} > \beta \mu_{i,k} \right]
$
which, if at most constant, would bound \emph{any} sequence of column indices and resolve the conjecture of Hemmer et al.~\cite{Hemmer16}.
However, since the events are binary, it is impossible to show strong concentration within a \emph{constant factor} of their expected values for \emph{each entry} of $C$ (e.g. $1/\avg[X_{i,i+1}]=\Omega(i)$).
Instead, our direct approach uses a special independence property for those sets of events that are given by a fixed sequence and that long index sequences are very unlikely feasible over $\Perm(S)$.

Let $I(n)$ be the set of all monotonic index sequences over $\{1,\ldots,n\}$ that start with $1$, i.e.
$$
I(n) := \Big\{ (b_0,b_1,\ldots,b_\ell) : 1=b_0 < b_1 < b_2 < \ldots < b_\ell \leq n \Big\}
\quad \text{, thus } |I(n)|=2^{n-1}~.
$$
For each sequence 
    $\sigma=(b_0,b_1,\ldots,b_\ell) \in I(n)$, let 
    ${\cal X}(\sigma):=\{X_{b_0,2},\ldots,X_{b_0,b_1},X_{b_1,b_1+1},\ldots,X_{b_\ell,n}\}$ 
be the $n-1$ events of the sequence.
Each sequence $\sigma=(b_0,b_1,\ldots,b_\ell)$ is associated to the two random variables $\sigma, \sigma' : \Perm(S) \to \N$ that map permutation $\pi$ to 
\begin{align*}
\sigma(\pi) :=	\sum_{Z \in \X(\sigma)} Z(\pi)      
    &&\text{and}&&  {\sigma'(\pi)}:= \sigma(\pi) - \sum_{i=1}^\ell X_{b_{i-1},b_i}(\pi)  
\quad.
\end{align*}
Note that $\avg[\sigma] \leq \delta H_n $ for any $\sigma \in I(n)$, due to Equation~(\ref{eq:bound-faces-of-sequence}).
Let $\beta>1$ be a sufficiently large constant, we define for each sequence $\sigma \in I(n)$ three events
\begin{align*}
F_\sigma(\pi)&\iff\textstyle\bigwedge_{(b_i,b_{i+1}) \in \sigma}X_{b_i,b_{i+1}}(\pi) 
&&
~&B_\sigma(\pi)\iff\sigma(\pi)\geq\beta\ln n~~~\\
~&~ ~&&~
~&B'_\sigma(\pi)\iff\sigma'(\pi)\geq\beta\ln n~.
\end{align*}
That is, $F_\sigma$ occurs if sequence $\sigma$ is feasible on the permutation $\pi$.
Note that $\avg[\sigma|F_\sigma] \geq |\sigma|$, 
for a sequence $\sigma=(b_0,b_1,\ldots,b_\ell)$ of length $|\sigma|:=\ell$.
To circumnavigate this dependency, we also use the events $B'_\sigma$  that drop the feasibility events from the summation and thus have $\Pr[B'_\sigma] \leq \Pr[B_\sigma]$.

\begin{lemma}
    For each sequence $\sigma \in I(n)$, the events in $\X(\sigma)$ are {\blue roughly independent, i.e.
    for each $Y \subseteq \X(\sigma)$ we have $\Pr[\bigwedge_{Z\in Y}Z]\leq \prod_{Z\in Y}\delta / \rho(Z)$, where $\rho(X_{i,j})=j$ denotes the row index of an event. }
\end{lemma}
{\blue
\begin{proof}  

    Let $X_{i,j} \in Y$ with $j$ minimal and $Y' = Y\setminus\{X_{i,j}\}$.
    To prove $\Pr[X_{i,j}, Y'] \leq (\delta/j) \Pr[Y'] $, we condition on the subsets $S' \in \binom{S}{j}$ and use that $\Pr[X_{i,j}~|~S_{\leq j}=S']\leq \delta/j$ for each $S'$, see~(\ref{eq:harmonic-events}).
    That is
\begin{align*}
    \Pr[X_{i,j},~Y'] 
&= 
    \frac{1}{\binom{n}{j} (n-j)!} \sum_{(s_{j+1},\ldots,s_n)} 
    \Pr[X_{i,j},~Y'~|~s_{j+1},\ldots,s_n ]
\\
\text{``Total Probability:''}\quad
&= 
    \frac{1}{\binom{n}{j} (n-j)!}
    \sum_{(s_{j+1},\ldots,s_n)} 
     \underbrace{\Pr[X_{i,j}~|~Y',~s_{j+1},\ldots,s_n]}_{= \Pr[X_{i,j}~|~s_{j+1},\ldots,s_n] ~\leq~ \delta/j} \cdot \Pr[Y'~|~s_{j+1},\ldots,s_n] 
\\
&\leq 
    \frac{(\delta/j)}{\binom{n}{j} (n-j)!} 
    \sum_{(s_{j+1},\ldots,s_n)} 
    \Pr[Y'~|~s_{j+1},\ldots,s_n] = (\delta/j)\Pr[Y']
      \quad,
\end{align*}
as required.
The bound uses that the choices in the sampling without replacement are independent.
Specifically, in any set $S' =S \setminus \{s_{j+1}, \ldots, s_n\}$, the choice of $s_j$ and $s_i$ are unaffected by adjacency to $s_{j+1}$, i.e. sampling the elements from $S'$ is regardless of them having color $A_{j+1}$ or $N_{j+1}$.
\end{proof}

From this lemma, we have for each sequence $\sigma=(b_0,\ldots,b_\ell) \in I(n)$ that
\begin{align}
    \Pr\left[F_{(b_0,\ldots,b_\ell)}\right]
    &\leq 
{ %
\textstyle\prod_{i=1}^\ell ( \delta / b_i )
}
    \label{eq:prob-feasible} \\
\avg \left[  \prod_{X_{i,j} \in {\cal X}(\sigma)} e^{X_{i,j}} \right]
&\leq (e^0 +e^1\delta/2) \avg \left[  \prod_{X_{i,j} \in {\cal X}(\sigma) \setminus\{X_{1,2}\}} e^{X_{i,j}}\right]
\leq (1+e\delta/2)\ldots(1+e\delta/n) \leq e^{e\delta H_n}
    \label{lem:whp-fixed-sequence}~,
\end{align}
where the last inequality applies the argument from the last proof in the Chernoff technique.
}

\begin{lemma} \label{lem:nof-feasible}
    The expected number of feasible sequences is at most polynomial in $n$.
\end{lemma}

\begin{proof}
Let $g(b_0,\ldots,b_\ell)=b_\ell$ denote the last element of a sequence {\blue and $f(b_0,\ldots,b_\ell)=\prod_{i=1}^\ell (\delta/b_i)$}.
We group the summation terms into those that have and haven't $g(\sigma)=n$, yielding
\begin{align*}
\sum_{\sigma \in I(n)} \Pr[F_\sigma]
&=
\sum_{\sigma \in I(n)~:~g(\sigma) = n} \Pr[F_\sigma]+
\sum_{\sigma \in I(n)~:~g(\sigma) < n} \Pr[F_\sigma] 
\\
\text{``Equation~(\ref{eq:prob-feasible}):''}\quad
&\leq
\sum_{\sigma \in I(n-1)} {\blue f(\sigma)} \cdot \frac{\delta}{n}+ \sum_{\sigma \in I(n-1)} {\blue f(\sigma)}
\\
&=
\left(\frac{\delta}{n} + 1\right)
\sum_{\sigma \in I(n-1)} {\blue f(\sigma)}
\quad = \quad \ldots \quad = 
\prod_{j=2}^{n}\left(\frac{\delta}{j} + 1 \right)
\\
&=
\prod_{j=2}^{n}\frac{j+\delta}{j} =\frac{ (n+1)(n+2)\ldots(n+\delta)}{(\delta+1)!} = \O(n^\delta)\quad,
\end{align*}
as required.
\end{proof}

\begin{observation} \label{obs:overcounting}
For summation over the integers $\{1,\ldots,n\}$, we have 
$\sum_{i_1 < i_2 < \ldots < i_\ell} \frac{1}{i_1 i_2 \ldots i_\ell} < H_n^\ell / \ell !$.
\end{observation}
\begin{proof}
To see that
$$
    \sum_{i_1 < i_2 < \ldots < i_\ell} \frac{1}{i_1 i_2 \ldots i_\ell} 
    = 
    \sum_{i_1} \sum_{i_2 \neq i_1 }
    \sum_{i_3 \not\in \{i_1,i_2\} } \ldots \sum_{i_l \not\in \{i_1,i_2,\ldots,i_{\ell-1}\} } \frac{1}{i_1 i_2 \ldots i_\ell} /\ell! 
$$
one considers a fixed set $A=\{i_1,\ldots,i_\ell\}$ that appears in the expression on the left side and observes that $A$ appears exactly $\ell!$ times as assignment to the $\ell$ indices on the right side.
The inequality follows, since the right side is less than
$\frac{1}{\ell!}
    \sum_{i_1}\sum_{i_2} \ldots \sum_{i_\ell} \frac{1}{i_1 i_2 \ldots i_\ell}
    =
    H_n^\ell / \ell!$~.
\end{proof}

We are now ready to show our second main result.
\begin{theorem}
There is a constant $\beta > 1$, such that the probability that any feasible column sequence exceeds $\beta \ln n$ is less than $1/n^{\Omega(\beta)}$.
\end{theorem}
\begin{proof}
Using the union bound, we have
$
\Pr\left[ \bigcup_{\sigma \in I(n)} B_\sigma \wedge F_\sigma \right] 
\leq
 \sum_{\ell=1}^{n-1} \sum_{\sigma:|\sigma|=\ell} \Pr[B_\sigma \wedge F_\sigma]$.

For long sequences, i.e. $\ell > \beta' H_n$, we use the following bound
\begin{align*}
\sum_{\sigma~:~|\sigma|=\ell} \Pr[B_\sigma \wedge F_\sigma]
\leq
&\sum_{\sigma~:~|\sigma|=\ell} \Pr[F_\sigma]
\\
\text{``Equation~(\ref{eq:prob-feasible}):''}\quad\quad
\leq
& \sum_{\sigma~:~|\sigma|=\ell} ~~\prod_{b_i \in \sigma~:~i>0 } \frac{\delta}{b_i}
\leq
\sum_{ i_1 < i_2 < \ldots < i_\ell  } \frac{\delta^\ell}{i_1 i_2 \ldots i_\ell} 
\\
\text{``Observation~\ref{obs:overcounting}:''}\quad \quad
<
&~\delta^\ell H_n^\ell /\ell!  
\\
\text{``Stirling:''}\quad \quad
=&~\O \left( \left(\frac{\delta H_n}{\ell/e}\right)^\ell \right)
 = \O \left( \left(\frac{\delta e}{\beta' } \right)^\ell \right) 
 \overset{\beta' > \delta e^2}{=}
 \O \left( 1/ n^{\beta'} \right)  ,
\end{align*}
which suffices for a union bound over the (less than $n$) possible values of $\ell$.

For short sequences, i.e. $\ell \leq \beta' H_n$, we show that a $\beta >\beta'$ is sufficient in the definition of the events $B_\sigma$.
{\blue
For any sequence $\sigma$ of length $\ell$, we have that $F_\sigma$ occurs if and only if $\sigma-\sigma'=\ell$.
Thus, for any $m \geq \ell$, we have $\Pr[\sigma \geq m ~|~ F_\sigma]=\Pr[\sigma \geq m ~|~ \sigma-\sigma'=\ell]=\Pr[\sigma' \geq m-\ell]\leq \Pr[\sigma' \geq m-\beta'\ln n]$.
Since $\Pr[B'_\sigma] \leq \Pr[B_\sigma]$ for any threshold $\beta$, we have from Eq.~(\ref{lem:whp-fixed-sequence}) the high probability bound $\Pr[\sigma \geq (\beta-\beta')\ln n] \leq e^{e\delta \ln n }/e^{(\beta-\beta')\ln n}=\frac{n^{e\delta}}{n^{\beta-\beta'}}$.
}
Hence the probability bound over all short sequences is at most
$$
\sum_{\sigma :|\sigma|\leq \beta' H_n}
\Pr[B_\sigma|F_\sigma] \Pr[F_\sigma]
\leq
\frac{n^{e\delta}}{n^{\beta-\beta' }}     \sum_{\sigma :|\sigma|\leq \beta' H_n} \Pr[F_{\sigma}] \leq 
\frac{n^{e\delta}}{n^{\beta-\beta' }}     \sum_{\sigma } \Pr[F_{\sigma}]
\leq
\O(n^{4\delta})/n^{\beta-\beta' } \quad,
$$
where the last inequality is due to Lemma~\ref{lem:nof-feasible}.
This is $1/n^{\Omega(\beta)}$ for a sufficiently large constant $\beta$.
\end{proof}

\section{Further Applications and Improved Results} \label{sec:conf-spaces}
We now discuss various extensions of the technique that yield improved bounds.

\subsection{TSD size for crossing segments} \label{sec:tsd-size-crossing}

From the results on Davenport-Schinzel sequences, we have that the number of $x$-intervals of the lower envelope of $m$ (potentially crossing) segments is $\O(m\alpha(m))$, where $\alpha$ is the inverse of Ackermann's function (see~\cite{AgerwalSharirDavSchinzelCG}).
Using $|P^+|\leq \O(\alpha(j))\sum_{i<j} X_{i,j}$ in the proof of Lemma~\ref{lem:relation-events-destroyed} yields that the number of occuring row events is in $[D_j/\O(\alpha(j)) , 3D_j]$.
Applying this lower bound however can only give estimates for the TSD size of the form $\O(\alpha(n)(n+k))$ with probability at least $1-1/e^{n+k}$, since $\alpha(n) \neq \O(1)$.

This shortcoming points out an interesting property of our Pairwise Event technique that we will strengthen in this section.
The underlying problem in obtaining a stronger bound is that the worst-case size of the TSD on $n$ segments is $\Omega\left(\alpha(\frac{n}{2})n^2\right)$, which exceeds the event count $\binom{n}{2}$ by more than a constant factor.
We overcome this obstacle by instead using $\binom{n}{3}$ events to analyze the total number of structural changes of the RIC on crossing segments.

Recall that $|\A(S)|=1+3n+3k$ (see Section~\ref{sec:recap-tsd}).
We partition set of trapezoidal faces ${\A(S)=: \A_\otimes(S) \cup \A_\odot(S)}$ into two classes,
those where $\textit{left}(\Delta)$ is a crossing and those where $\textit{left}(\Delta)$ is a segment endpoint (or the left point of the domain boundary).
This partition extends to the set of incident faces 
$f(s_j, S_{\leq j}) =: f_\otimes (s_j, S_{\leq j}) \cup f_\odot (s_j, S_{\leq j}) $ 
and we define the random variables $E_j, O_j, C_j : \Perm(S) \to \N$ such that
$$
D_j = 
\underbrace{|f_\otimes (s_j, S_{\leq j})|}_{=:O_j+C_j}
+
\underbrace{|f_\odot(s_j, S_{\leq j})|}_{=:E_j} 
~,
$$
where $C_j$ is the number of trapezoids whose $\textit{left}$ is due to a crossing of $s_j$ and $O_j$ the remainder of this class.
Note that $3k = \sum_j C_i(\pi)$ regardless of the permutation $\pi \in \Perm(S)$, and we have
$3k + {\normalfont\avg}[\sum_j O_j] =\O(n+k)$.
We define the events
\begin{align*}
X_{i,j} &\Longleftrightarrow f_\odot(s_j, S_{\leq j}) \text{ contains a trapezoid bounded by } s_i
&& 1\leq i<j\leq n
\\
Y_{h,i,j} &\Longleftrightarrow  f_\otimes(s_j, S_{\leq j}) \text{ contains a trapezoid with crossing }(s_h,s_i)
&& 1 \leq h < i < j \leq n
~.
\end{align*}

The proof of Lemma~\ref{lem:relation-events-destroyed} gives that 
$E_j / 6 \leq \sum_{i<j}X_{i,j} \leq 3E_j$ for every $\pi \in \Perm(S)$ and $j>1$.
Since $\avg[E_j] \leq \avg[D_j] \leq 4\gamma$, we have $\sum_{i,j} \avg[X_{i,j}] \leq 12\gamma n$.
Thus Theorem~\ref{thm:inductive-chernoff} gives for any $T>0$ the bound $\Pr[\sum_j E_j \geq T/6] \leq \exp(12\gamma n - T \ln2 )$.

It remains to bound $\Pr[\sum_j O_j \geq T]$.
Let the set $r(j):=\{ Y_{h,i,j}: h<i<j\}$ be the events of row $j$ and let the set $c(h):=\{Y_{h,i,j}:h<i<j\}$ be the events of column $h$.
We now show that $O_j$ coincides with the number of crossings visible from $s_j$ in $\A(S_{\leq j})$.
\begin{lemma} For each $\pi \in \Perm(S)$ and $j>1$, we have $O_j (\pi) = \sum_{Y \in r(j)} Y(\pi)$.
\end{lemma}
\begin{proof}
Let $S'= S_{\leq j}$. 
From the definition, $O_j = |f_\otimes(s_j,S')|-C_j$ counts only those trapezoids of $\A_\otimes(S')$ that are incident to $s_j$ but due to a crossing of some $\{s,s'\} \subseteq S'\setminus\{s_j\}$.
For ``$\leq$'', observe that only one of the three trapezoids due to a crossing $\{s,s'\}$ can be incident to $s_j$ and each such trapezoid has one occurring row event. 
For ``$\geq $'', we use that crossing line segments $\{s,s'\}$ that are adjacent to $s_j$  have at most one intersection point that is either above or below $s_j$ in $\A(S')$, each of these trapezoids is counted exactly once in $O_j$.
\end{proof}

\begin{theorem}
The upper-tail $\Pr[\sum_{h<i<j} Y_{h,i,j} \geq T ] \leq \exp( {\normalfont\avg}[\sum_{h<i<j} Y_{h,i,j}] - T \ln 2)$ for all $T>0$.
\end{theorem}
\begin{proof}
To use the proof technique of Theorem~\ref{thm:inductive-chernoff}, it is sufficient to have the analogues of Eq.~(\ref{eq:conditioning-uninformative}), (\ref{eq:deter-events-c1}), and (\ref{eq:deter-events-ci}).
We have, for each event $Y \in c(h)$ and event $Y' \in c(h')$ with $h'>h$, that
\begin{align*}
    \avg[Y|s_{h+1},\ldots,s_n] &= \avg[Y|Y',s_{h+1},\ldots,s_n]~.
\end{align*}
To see this, we consider the crossings of $s_j$ in $S_{\leq j}$ and (again) think of the backward process that builds the random permutation by successively choosing one of the remaining elements.
Picking $s_j$ fixes a partition of $S_{\leq j}=\{s_j\} \cup A_j \cup N_j$ into those segments that are adjacent to a crossing of $s_j$ and those segments that are not adjacent.
Hence picking $s_h$ from given $S \setminus\{s_{h+1},\ldots,s_n\}$ determines the outcome of \emph{all events} in $c(h)$, i.e. $Y_{h,i,j}$ occurs if and only if $s_h$ has color $A_i$ and $A_j$.

Hence for given suffix condition, the events in $c(h)$ are either certain or impossible and thus
\begin{align*}
\avg[ \prod_{Y \in c(1)}  \exp(tY)|s_2,\ldots,s_n] &= \prod_{Y \in c(1)} \avg[ \exp(tY)|s_2,\ldots,s_n] 
\\
\avg[ \prod_{Y \in c(h)}  \exp(tY)|s_h,\ldots,s_n] &= \prod_{Y \in c(h)} \avg[ \exp(tY)|s_h,\ldots,s_n] 
&&\forall~ 2 \leq h \leq n-2 ~.
\end{align*}
The remaining arguments of the proof of Theorem~\ref{thm:inductive-chernoff} require no modification, and consequently this proves the theorem.
\end{proof}

We conclude
\begin{align*} 
\Pr\left[\textstyle\sum_j E_j + O_j + C_j \geq 2T\right] &\leq
\Pr\left[\textstyle\sum_j E_j \geq T\right] + \Pr\left[\textstyle\sum_j B_j+C_j \geq T\right] \\
&\leq \exp\left(12\gamma n - T\tfrac{\ln 2}{6}\right) + 
      \Pr\left[\textstyle\sum_j B_j \geq T-3k \right] \\
&\leq \exp\left(12\gamma n - T\tfrac{\ln 2}{6}\right) + 
      \exp\left( \avg[ {\textstyle \sum} Y_{h,i,j}] - (T-3k)\ln 2 \right) \\
      &\leq 1/e^{n+k}~,
\end{align*}
where the last inequality follows from ${\normalfont\avg}[\sum Y_{h,i,j}]={\normalfont\avg}[\sum_j O_j]=\O(n+k)$ and choosing $T=\beta(n+k)$, with a sufficiently large constant $\beta>1$.
We have shown the following:

\begin{corollary}
    The TSD size is $\O(n+k)$ with probability at least $1-1/e^{n+k}$.
\end{corollary}

\subsection{History DAGs for 2D Delaunay Triangulation and 3D Convex Hulls} \label{ssec:dag-triangulations}
We briefly outline the RIC to define the necessary terminology.
For the input set $S$ of $n$ sites, we compute the Delaunay triangulation $\A(S)$ by inserting the sites, one at a time, to derive $\A(S_{j})$ from $\A(S_{j-1})$.
The initially empty triangulation $\A(S_0)$ only contains the bounding triangle $(-\infty,-\infty)(+\infty,-\infty)(0,+\infty)$.
Given the triangle $pqr$ in $\A(S_{j-1})$ that contains the next site $s_j$, we split the face into three triangles, i.e. $pqs_j, qrs_j$ and $rps_j$, and scan the list of incident triangles of $s_j$, in CCW order.
If one such triangle is not (local) Delaunay, we flip the edge and replace its former entry with both, now incident, triangles in the CCW list until all incident triangles are Delaunay.
The work is proportional to the degree of $s_j$ in the triangulation of $\A(S_{j})$, which is expected $\O(1)$ due to standard Backward Analysis.

For locating the face that contains $s_j$, there are \emph{two} well known search DAG variants that use the triangulation history $\A(S_0),\A(S_1),\ldots,\A(S_{j-1})$.
The first method~\cite[Section~$3$]{GuibasKS92} keeps record of all intermediary (non-Delaunay) triangles to search for next site $s_j$ (e.g. \cite[Chapter~$9.3$]{BergCompGeo}).
Instead of keeping intermediary triangles, the second method~\cite[Section~$5$]{GuibasKS92} simply keeps the final CCW lists, referenced by the deleted triangles, which allows point location searches to descend in the history using (radial) binary searches (e.g. \cite[Chapter~3.3]{MulmuleyBook}).
To simplify the presentation, we assume that the final CCW lists are stored as array for the binary search. (Skip lists can be used for the pointer machine model.)

It is well known that the radial-search method can also be used for the RIC of $3$D~Convex Hulls, with minor changes of the algorithm and analysis (see Figure~\ref{fig:dt-ch-dag}; cf. \cite{MulmuleyBook}).
The history size of either method is expected $\O(n)$ and the first method has expected $\O(n \log n)$ runtime (see e.g. \cite[Chapter~$9.5$]{BergCompGeo}).
However, high probability bounds for the history size are unknown, the best bound is Corollary~$26$ in~\cite{ClarksonMS93}, as discussed in Section~\ref{sec:rel-work}.
Moreover, the first method does not yield any bound for general point location queries (only for the input sites).
This is overcome by the second method, that has w.h.p.~$\O(\log^2 n)$ search cost for \emph{all} query points (see \cite[Theorem~$3.3.2$]{MulmuleyBook}). 

We now improve the tail bounds for the DAG size and show that point location cost is w.h.p. $\O(\log n)$, using the proposed technique of Pairwise Events.

\begin{figure}[t] \centering
    \includegraphics[]{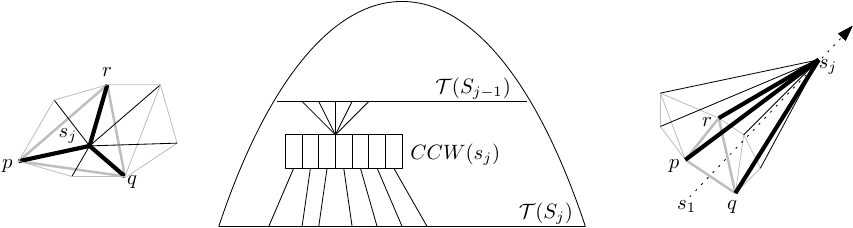}
    \caption{Schema of a search DAG using the radial-search for 2D Delaunay triangulations (left) and 3D convex hulls (right).}
    \label{fig:dt-ch-dag} 
\end{figure}

\subsubsection{Improved Space Bounds}
To simplify presentation, we use the ordinary non-degeneracy assumptions that no four points are co-circular (Delaunay) and no four points are co-planar (i.e. the faces of the $3$D convex hull are triangles).
Let the random variable $D_j$ be the number of triangles in $\A(S_j)$ that are incident to $s_j$.
Recall that $\avg[D_j]=\O(1)$ and $\sum_j \avg[D_j]=\O(n)$.
We define event 
$X_{i,j}$ to occur if and only if $s_i$ and $s_j$ are adjacent in $\A(S_j)$, i.e. they form an edge of the triangulation.
The analogue of Lemma~\ref{lem:relation-events-destroyed} is even simpler, since $D_j = \sum_{i<j} X_{i,j}$ holds with equality (i.e. $D_j$ is both upper and lower bound for the number of occurring row events).
Using Theorem~\ref{thm:inductive-chernoff} with these bounds gives the following result.

\begin{corollary}
    In the Randomized Incremental Construction of $2$D Delaunay Triangulations and $3$D Convex Hulls, the history size is $\O(n)$ with probability at least $1-1/e^n$.
\end{corollary}

\subsubsection{Improved Point Location Bound}
To improve the query time bound of the point location we need to modify our arguments of Section~\ref{sec:path-length}.
We discuss the, more technical, radial-search method.
The boundary priority of a triangle $\Delta=pqr$ is $b(\Delta):=\max\{\pi(p),\pi(q),\pi(r)\}$.
Since a sequence $\Delta_i$ of (Delaunay) triangles on a root-to-leaf path of the history DAG has strictly increasing $b(\Delta_i)$, we will again argue over all possible monotonic index sequences.
Let $\sigma=(b_0,\ldots,b_\ell)$ be a monotonic sequence of index values.
The search cost for a DAG path with index sequence $\sigma$ is at most
\begin{align}   \label{eq:radial-cost}
\sum_{i=1}^\ell \underbrace{X_{b_{i-1},b_i} \cdot \Big\lfloor 1+ \log_2 \big( D_{b_i} \big) \Big\rfloor}_{\blue := Y_i} ~.
\end{align}
We have 
$\avg[X_{i,j} \cdot D_j]=\avg[X_{i,j}] \avg[ D_j | X_{i,j}]=\avg[X_{i,j}] \avg[ D_j ]$, since the choice of $s_i$ from $S_j\setminus \{s_j\}$ does not depend on which elements are `adjacent' to $s_j$.
{ \blue
To prove a high probability bound for the cost~(\ref{eq:radial-cost}), we argue as in the proof for Equation~\ref{lem:whp-fixed-sequence}.
Conditioning $\avg \left[e^{t Y_1 }\right]$ on a suffix, we have for $t:= \ln 2$ that
\begin{align*}
&\avg\left[~\exp \big( t X_{1,b_1} \lfloor 1 + \log_2 D_{b_1} \rfloor \big)~\big|~s_{b_1+1},\ldots, s_{n}~\right]
\\
\leq 
&1 \cdot e^0 +  
\Pr[~X_{1,b_1}~|~s_{b_1+1},\ldots, s_{n}~] \cdot 
\avg \left[~\exp( t (1+ (\ln D_{b_1})/\ln 2)) ~\big|~X_{1,b_1}, s_{b_1+1},\ldots, s_{n}~\right]
\\
= &1 + \Pr[~X_{1,b_1}~|~s_{b_1+1},\ldots, s_{n}~] \cdot 2 \cdot \avg\left[~D_{b_1} ~\big|~X_{1,b_1}, s_{b_1+1},\ldots, s_{n}~\right]
=1+\O(1)/b_1
\quad.
\end{align*}
Thus $\avg[~\exp ( tY_1 + \ldots + tY_\ell)~] \leq {\textstyle \prod_{i=1}^\ell (1+\O(1)/b_i) \leq \exp  \sum_{i=1}^\ell \O(1)/b_i }\leq e^{\O( \ln n)} = n^{\O(1)}$ for any fixed sequence $\sigma$.
}
The remaining arguments in Section~\ref{sec:path-length} require no modification, beside using an appropriate constant $\gamma$, such that $|\A(S)|\leq \gamma n$, for triangulations. 
As a result we get:

\begin{corollary}
    Worst-case point location cost in the History of $2$D Delaunay Triangulations and $3$D Convex Hulls is w.h.p. $\O(\log n)$.
\end{corollary}

To obtain a Las Vegas verifier, we now describe a simple addition to the RIC that tracks this upper bound (on the radial search costs) as weighted path lengths in the History DAG over the course of the construction, i.e. from $\A(S_{j-1})$ to $\A(S_j)$.
We associate the weight $\lfloor\log_2 D_j \rfloor+1$ to every outgoing pointer of the CCW list of $s_j$ to the incident triangles of $\A(S_j)$.

For computing the maximum cost of an edge weighted path to the new triangles in $\A(S_j)$, we simply propagate the maximum cost from those triangles $\Delta$ that are deleted from $\A(S_{j-1})$ in the course of triangle replacements that create the CCW list, i.e. the replacements $(\Delta''', \Delta'')$ of the entry $\Delta'$ store the maximum of $\Delta$ and $\Delta'$.
Eventually, aforementioned weight of $\lfloor\log_2 D_j \rfloor+1$ is added to value of each of the triangles in the final CCW list of $s_j$.

\begin{corollary}
    Point location in the History DAGs of $2$D Delaunay Triangulations and $3$D Convex Hulls can be made worst-case $\O(\log n)$ by merely increasing the expected construction time by an additive constant.
\end{corollary}

\subsection*{Conclusion and Future Work}
We introduced a simple analysis technique that gives improved tail estimates for the size of the search structures from RIC. %
High probability bounds, though of general interest for dynamic maintenance, were unknown.
Consequently, we provided the insight that history-based RIC of the trapezoidation gives \emph{with high probability} a search DAG of $\O(n)$ size and $\O(\log n)$ query time in $\O(n \log n)$ time.
It has been shown recently that conflict-graph RIC cannot achieve the time bound in the worst-case.

Our technique also gives novel, high probability bounds for the combinatorial path length, which eluded the geometric backward analysis technique.
This allowed us to prove a conjecture on the depth of the search DAG for point location in planar subdivisions and a long-standing conjecture on the query time in the search DAGs for 2D nearest-neighbor search and extreme point search in 3D convex hulls.
Consequently, we provided the insight that identical high probability bounds on size, query time, and construction time, hold for these search DAGs.

The new algorithms we obtain from this are \emph{remarkably simple} Las Vegas verifiers for history-based RIC that give worst-case optimal DAGs, i.e. linear size \emph{and} logarithmic search cost.

For future work, we are interested in a finer calibration of the constants (e.g. if only polynomial decay is needed) and extensions of the technique for other history-based RICs with super-quadratic worst-case size (e.g. Delaunay tessellations and convex hulls in higher dimensions).

\subsection*{Acknowledgments}
The authors want to thank Sasha Rubin for collaborating on Observation~\ref{obs:overcounting}, Wolfgang Mulzer, Yanheng Wang and Daniel Zhang and for pointing out mistakes, Boris Aronov for discussions during his stay, Daniel Bahrdt for the github project \verb-OsmGraphCreator-, and Raimund Seidel for sharing his excellent lecture notes on a CG course he thought 1991 at UC Berkeley.

\newpage
\bibliographystyle{alpha}
\bibliography{references}

\end{document}